%% file: arxiv.tex
\documentclass[letterpaper, 11pt]{article}

\usepackage[english]{babel}
\usepackage[utf8x]{inputenc}
\usepackage[T1]{fontenc}

\usepackage[round]{natbib}   

\usepackage[margin=1in]{geometry}
\usepackage{graphicx}
\usepackage{subfigure}
\usepackage{booktabs} 

\usepackage{hyperref}

\usepackage{algorithm}
\usepackage{algorithmic}

\usepackage{empheq}
\usepackage{amsfonts}
\usepackage{amsthm}
\usepackage{bbm}
\usepackage{authblk}

\newtheorem{definition}{Definition}

\newtheorem{theorem}{Theorem}
\newtheorem{lemma}{Lemma}

\allowdisplaybreaks
\renewcommand{\hat}{\widehat}
\usepackage{enumitem}

\usepackage{xspace}

\newcommand{\E}{\mathbb{E}}

\newcommand{\D}{\mathcal{D}}

\title{Learning Ising Models with Independent Failures}

\author[1]{Surbhi Goel\footnote{surbhi@cs.utexas.edu}}
\author[2]{Daniel M. Kane\footnote{dakane@ucsd.edu}}
\author[1]{Adam R. Klivans\footnote{klivans@cs.utexas.edu}}
\affil[1]{Department of Computer Science, University of Texas at Austin}
\affil[2]{Department of Computer Science, UCSD}
\date{}

\begin{document}

\maketitle

\begin{abstract}%
  We give the first efficient algorithm for learning the structure of an Ising model that tolerates independent failures; that is, each entry of the observed sample is missing with some unknown probability $p$.  Our algorithm matches the essentially optimal runtime and sample complexity bounds of recent work for learning Ising models due to \cite{klivans2017learning}.

  We devise a novel unbiased estimator for the gradient of the \textit{Interaction Screening Objective} (ISO) due to \cite{vuffray2016interaction} and apply a stochastic multiplicative gradient descent algorithm to minimize this objective.  Solutions to this minimization recover the neighborhood information of the underlying Ising model on a node by node basis.

 \end{abstract}


\section{Introduction}
\input{Files/intro}

\section{Preliminaries}
\input{Files/prelims}

\section{Main Approach}
To recover the underlying graph, we will show how to reconstruct the neighborhood of each vertex by solving a convex constrained minimization problem. We will first describe the optimization problem and then show how to minimize the same using a stochastic first-order method assuming access to an unbiased estimator of the gradient. Subsequently we will show how to recover the neighborhood from the optimization solution. Lastly, we will detail the construction of the unbiased estimators for the two given noise models.

\noindent\textbf{Note:} For ease of presentation, we will consider Ising models with zero mean-field ($\theta = 0$). For details on the non-zero mean-field case, we refer the reader to Appendix \ref{app:nonzero}

 \input{Files/optimization}

\input{Files/estimation}
\section{Main Result}
\input{Files/maintheorem}

\section{Extension to unknown \boldmath$p$}
\input{Files/unknownp}


\section{Conclusions and Open Problems}
Our result highlights the importance of choosing the right surrogate
loss to obtain noise-tolerant algorithms for learning the structure of
Ising models.  It would be interesting to know if other regression-based algorithms
(e.g., the Sparsitron due to \cite{klivans2017learning}) can learn
with independent failures.  Other open problems include learning with
missing data for distinct and {\em unknown} error rates
$p_{1},\ldots,p_{n}$ (we can only handle the unknown error-rate case when all
$p_{i}$ are equal) and improving the dependence on $p$ in the
sample complexity bound.


\bibliographystyle{plainnat}
\bibliography{references}

\appendix
\input{Files/appendix}

\end{document}

%% file: Files/intro.tex
Ising models are fundamental undirected binary graphical models that capture pair-wise dependencies between the input variables. They are well-studied in the literature and have applications in a large number of areas such as physics, computer vision and statistics (\cite{jaimovich2006towards,koller2009probabilistic,choi2010exploiting,marbach2012wisdom}). One of the core problems in understanding graphical models is \textit{structure learning}, that is, recovering the structure of the underlying graph given access to random samples from the distribution.  Developing efficient algorithms for structure learning is a heavily-studied topic, especially for the case when the underlying graph is sparse or has bounded degree (\cite{dasgupta1999learning,lee2007efficient,bresler2008reconstruction,ravikumar2010high,bresler2014hardness,bresler2015efficiently,vuffray2016interaction,hamilton2017information,klivans2017learning,wu2018sparse}).  \cite{klivans2017learning} were the first to give an algorithm for learning the structure of Ising models (and more generally higher order MRFs) with essentially optimal runtime and sample complexity.

The focus of this paper is the setting where the samples drawn from the Ising model are corrupted by noise.  More specifically, we consider the \textit{independent failure model} where each entry of each sample is independently corrupted -- missing or flipped --- with some constant possibly unknown probability $p$.   Such a scenario can be expected, for example, in a sensor network, where sensors fail to report data occasionally due to internal failures.  In addition to being a natural question, structure learning in this model has been specifically posed as an open problem by \cite{chenlearning}. 

\subsection{Our Result}
The main contribution of our paper is a simple stochastic algorithm for learning the structure of an Ising model in the presence of corruptions with near optimal sample and time complexity:
\begin{theorem}[Informal version]
Consider an Ising model with dependency graph $G$ over $n$ vertices with minimum edge weight $\beta$. Further assume that the sum of absolute values of the weights of outgoing edges from each vertex is bounded by $\lambda$. Given corrupted draws from the Ising model such that each entry is missing or flipped with known probability $p \in (0,1)$, there exists an algorithm that recovers the underlying structure in time $\widetilde{O}(n^2)$ using at most $\lambda^4 \beta^{-4}\log (n) \exp(C \lambda)$ samples where $C$\footnote{In the missing-data setting, $C$ scales as $\frac{1}{1-p}$.  For $p<1/2$, $C$ is at most 10. We refer the reader to the proof of Theorem \ref{thm:main} for the exact dependence on $p$.} depends on $p$. 
\end{theorem}
For missing data, we extend our approach to the setting where $p$ is unknown. We show that using only $O(\log n)$ fresh samples to estimate these probabilities suffices for the guarantees to hold. Our results can also be easily extended to other similar noise models such as independent block failures where instead of each entry being flipped, fixed subsets of entries are simultaneously missing or flipped. 
Our work is the first efficient algorithm for learning the structure of Ising models in the presence of independent failures; both our running time and sample complexity are essentially optimal (matching recent work due to \cite{klivans2017learning}). 

\subsection{Our Approach} 

Our approach has three main components that we briefly explain here.
\paragraph{Optimization Problem:} We follow the ``nodewise-regression'' approach of recovering the graph by solving an optimization problem to identify the neighborhood of each vertex. To do so, we minimize the Interaction Screening Objective (ISO) proposed by \cite{vuffray2016interaction} over the $l_1$-norm constrained ball.  The ISO satisfies a property known as \textit{restricted strong convexity} which allows us to recover the underlying edge weights from an approximate minimum solution.
\paragraph{Minimization Procedure:} Instead of using convex programming as in \cite{vuffray2016interaction}, we minimize ISO using \textit{stochastic multiplicative gradient descent} (SMG) due to \cite{kakade2008}.  SMG is a stochastic, multiplicative-weight update algorithm and therefore runs on the simplex (modeling a probability distribution).  As such, we need to transform our optimization problem to the simplex. The main motivation for using SMG is two-fold: 1) it can handle $l_1$-constraints, 2) it requires access to only an unbiased, bounded estimator of the gradient to give convergence guarantees.
\paragraph{Unbiased Gradient Construction:} Despite independent failures in our sample, we are able to construct an unbiased estimator of the gradient of the (modified) ISO on the simplex. Our constructor is able to exploit the \textit{decomposability} of the ISO, as it is an exponential of a linear function.  This allows us to use the independence property of the failures to come up with a simple unbiased estimator. The techniques used for this construction may be of independent interest.

\subsection{Related Work}
\paragraph{Ising Models.} Structure learning for Markov Random Fields has been studied since the 1960s.  For example, \cite{chow1968approximating} gave a greedy algorithm for undirected graphical models assuming the underlying graph is a tree.  There have been many works for learning Ising models under various assumptions on the structure of the underlying graph (e.g., \cite{lee2007efficient,yuan2007model,ravikumar2010high,yang2012graphical}). For Ising models specifically, the first assumption-free result (that is, no assumptions are made on the underlying graph other than sparsity) was given by \cite{bresler2015efficiently}.  One drawback of Bresler's work, however, is that the sample complexity has a doubly exponential dependence on the sparsity of the underlying graph. This was improved by \cite{vuffray2016interaction} where they used general-purpose tools for convex programming to minimize a certain objective function, but the running time of their approach was suboptimal. Subsequently, \cite{klivans2017learning} gave a multiplicative weight update algorithm called \textit{Sparsitron} that achieved near-optimal sample complexity (essentially matching known information-theoretic lower bounds) and near-optimal running time (under a computational hardness assumption for the {\it light-bulb problem} \cite{Val15}).  Recently, \cite{wu2018sparse} gave a different proof of \cite{klivans2017learning} using $l_1$-regularized logistic regression.

\paragraph{Learning Ising Models with Noise.} 
None of the works mentioned above handle the setting where there is noise in the observed data.  The problem of structure learning in the independent failure model (specifically missing data) was raised by \cite{chenlearning}.  The only work we are aware of that obtains positive results in this model is due to \cite{hamilton2017information}, who applied a broad generalization of the work of \cite{bresler2015efficiently}.  Their sample complexity, however, has a doubly exponential dependence on the sparsity of the underlying graph.  Our result gives a singly exponential dependence on the sparsity (known to be information-theoretically optimal) and can handle non-sparse graphs as long as they have small $l_1$-norm.

\paragraph{Robust Learning of Ising Models.} 
\cite{lindgren2018robust} studied structure learning in a different noise model motivated by recent results in robust learning.  In their work, an adversary is allowed to arbitrarily corrupt some $\eta$ fraction of a training set drawn from the underlying distribution.   They showed that if an adversary is allowed to corrupt even an exponentially small fraction of the samples, then no algorithm is robust. They further showed that their bounds are tight by proving robustness of \textit{Sparsitron} against exponentially small adversarial corruption.  The adversarial model is incomparable to the model studied here: in the independent failure model, {\em every} example will (on average) have a $p$ fraction of the entries missing (or flipped). 


\paragraph{Learning with Missing Data.} For general learning problems, various methods have been proposed to handle missing data including heuristics and maximum likelihood methods. In the context of high-dimensional sparse {\em linear} regression, \cite{loh2011high} proposed simple estimators to handle missing data based on solving optimization problems, and further showed that simple gradient descent recovers close to optimum solution despite non-convexity.  For distribution learning, \cite{shah2018learning} recently gave the first positive results for learning mixtures of gaussians with missing data with optimal sample/runtime complexity. It is not immediately clear how to use either of these techniques for our problem.

\subsection{Notation}
For vector $x \in \mathbb{R}^n$, $x_{-i}$ denotes $(x_j: j \ne i)$ and $x_S$ denotes $(x_i: i \in S)$. $|| \cdot ||_p$ denoted the $l_p$-norm. We denote the different distance/divergence metrics for distributions as follows: $\mathsf{KL}(P||Q)$ denotes the Kullback-Lieber divergence between probability distributions $P$ and $Q$, and $d_{\mathsf{TV}}(P,Q)$ denotes the total variation distance between $P$ and $Q$. We denote the $l_1$ ball of radius $R$ using $\mathsf{B}(R, n):= \{x \in \mathbb{R}^n: ||x||_1 \le R\}$ and the simplex with radius $R$ using $\Delta(R, n):= \{x \in \mathbb{R}^n: x \ge 0, \sum_{i=1}^n X_i = R\}$.

%% file: Files/prelims.tex
\begin{definition}[Ising model]
Let $A\in \mathbb{R}^{n \times n}$ be a symmetric matrix with $A_{ii} = 0$ for all $i$ and $\theta \in \mathbb{R}^n$ be the mean-field vector. Let $G=(V,E)$ be an undirected graph with $V = [n]$ such that $(i,j) \in E$ if and only if $A_{ij} \ne 0$. The $n$-variable Ising model with underlying dependency graph $G$ is a distribution $\D(A, \theta)$ defined on $\{-1,1\}^n$ such that
\[
\Pr[Z = z] = \frac{1}{\mathcal{Z}} \exp\left(\sum_{(i,j) \in E}A_{ij} z_iz_j + \sum_{i \in V}\theta_i z_i\right)
\]
where $\mathcal{Z} = \sum_{z \in \{-1,1\}^n}\exp\left(\sum_{(i,j) \in E}A_{ij} z_iz_j + \sum_{i \in V}\theta_i z_i\right)$ is the normalizing factor. We denote the minimum edge weight by $\beta(A):= \min_{(i,j) \in E} |A_{ij}|$ and the width of the model by $\lambda(A, \theta):=$ $\max_{i \in V}\left(\sum_{j| (i,j) \in E}|A_{ij}| + |\theta_i|\right)$. 
\end{definition}
We will assume that the minimum edge weight is at least $\beta \le \beta(A, \theta)$ and the model is width bounded by $\lambda \ge \lambda(A, \theta)$. For ease of presentation, we will suppress notation and denote $\D(A, \theta)$ by $\D$.
A useful property of bounded width Ising models is that the conditional distributions are bounded away from 0 and 1. More formally,
\begin{lemma}[\cite{bresler2015efficiently}]\label{lem:bresler}
For any node $v \in V$, subset $S \subseteq V \backslash \{v\}$, and any fixed configuration $z_S$ of the subset $S$,
\[
\min\left\{\Pr[Z_{v} = 1|Z_S = z_S], \Pr[Z_{v} = -1|Z_S = z_S] \right\} \ge\frac{1}{2}\exp(-2 \lambda).
\]
\end{lemma}

In this paper, we are interested in \textit{structure learning},
that is, recovering the edges of the dependency graph of an unknown Ising model given independent draws from it. We will focus on recovery in the presence of corruptions in the observed samples. The model of corruptions we study in our paper is known as the \textit{independent failures models}: corruptions of each variable are independent of all the other variables. More specifically, we will consider the following two special cases of independent failures:
\begin{enumerate}
  \item \textit{Missing Data}: In this setting, for sample $z$ we will instead observe $x$ such that for all $i$, with probability $p_i$, $x_i = ?$ (missing) otherwise $x_i = z_i$.
  \item \textit{Flipped Data}: In this setting, for sample $ z$ we will instead observe $x$ such that for all $i$, with probability $p_i$, $x_i = -z_i$ (flipped) otherwise $x_i = z_i$.
\end{enumerate}

%% file: Files/optimization.tex
\subsection{Optimization Problem}
 Consider the neighborhood of a fixed vertex (WLOG we choose $n$, we can do the same analysis for any vertex). \cite{vuffray2016interaction} proposed to minimize the Interaction Screening Objective (ISO) as follows:
\begin{align*}
\noalign{\textbf{Optimization Problem 1:}}
\min_{v \in \mathbb{R}^{n-1}} & \quad\quad S(v) := \mathbb{E}_{Z \sim \D} \left[\exp\left(-\sum_{j=1}^{n-1}v_{j} Z_n Z_j\right) \right]\\
\text{subject to} &\quad\quad || v||_1 \leq \lambda.
\end{align*}

\sloppy
\cite{vuffray2016interaction} studied the empirical version of Optimization Problem 1: the objective is computed using $m$ samples $\{z^1, \ldots, z^m\}$ drawn from the Ising model as $\hat{S}(v):= \frac{1}{m} \sum_{i=1}^m \exp\left(-\sum_{j=1}^{n-1}v_{j} z^i_n z^i_j\right)$. They proved various useful properties of $\hat{S}$ that directly extend to $S$. Firstly, they showed that the optimal solution to Optimization Problem 1 captures the edge weights of the neighborhood of $n$ in $G$. More formally,

\begin{lemma}[\cite{vuffray2016interaction}] \label{lem:grad}
Let $v^* \in \mathbb{R}^{n-1}$ be such that for $i \in [n-1]$, $v^*_i = A_{ni}$ if $(i,n)\in E$ else 0. Then we have, $\nabla S(v^*) = 0$ and $v^*$ is a global minimum Optimization Problem 1.
\end{lemma}

Further they proved that $S$ satisfies a
property known as \textit{restricted strong convexity} (RSC) which
enables us to recover a vector close to $v^*$ by approximately minimizing the objective. Here we give a
stronger version of their result which holds more generally. 
\begin{lemma}\label{lem:rsc}[RSC for $S$]
For all $ v \in \mathsf{B}(\lambda, n-1)$,
\[
S(v) - S(v^*) \ge \nabla S(v^*)^T(v - v^*) + \frac{\exp(-3\lambda)}{1 + \lambda} || v - v^*||_\infty^2
\]
\end{lemma}
\noindent \textbf{Note:} Our proof technique improves on their result,
as we work directly with the $l_\infty$ norm, avoiding the use of the $l_2$ norm in their proof (c.f. Lemma 7 and 8 \cite{vuffray2016interaction}). Using our analysis in their proof improves their sample complexity bounds for structure learning from $\max(d, \beta^{-2})d^3 e^{O(\beta d)}\log n$ to $\max(d^2, \beta^{-2})d^2 e^{O(\beta d)}\log n$ where $d$ is the maximum degree of each node in $G$. This improvement is highlighted when $\beta d < 1$ (where the improvement is by a factor of $d$).

With the above lemma in hand, we can show that small loss implies closeness to $v^*$ and hence recovery of the edge weights of the neighborhood of $n$. More formally, if $S(v) - S(v^*) \le \epsilon$ then $||v - v^*||_\infty^2 = \max_{i \in [n-1]} |v_i - A_{ni}|^2 \le (1 + \lambda) \exp(3 \lambda) \epsilon$.

\subsection{Minimizing ISO with \boldmath$l_1$ Constraint}
 To solve Optimization Problem 1,
    we will use an algorithm due to \cite{kakade2008} called Stochastic
 Multiplicative Gradient Descent (SMG) (see Algorithm
 \ref{alg:smg}). SMG is a multiplicative-weight update algorithm that
 optimizes over the simplex instead of the $l_1$-constraint ball.  As
 such, we will need to reduce our problem from the $l_1$ ball to the
 simplex to obtain the appropriate guarantees.   Since the SMG analysis only
 appears in a set of lecture notes (and there are some minor errors in
 the writeup), for completeness we include the proof of the following theorem in the appendix:
  \begin{algorithm}\label{alg:smg}
  \caption{Stochastic Multiplicative Gradient Descent}
\begin{algorithmic}
\STATE \textbf{Input:} $l_1$-bound $W$, dimension $k$ and convex function $c$ over $\Delta(W,k)$.
  \STATE Set $u^1_i := W/k$ for all $i$.
  \FOR{$t=1$ {\bfseries to} $T$}
  \STATE For unbiased estimate of gradient $ g^t$ of $c(u^t)$ and $C^t := \frac{(u^t)^T g^t}{W}$
  \STATE For all $i$, update $u_i^{t+1} := u_i^t\left(1 - \eta g^t_i + \eta C^t \right)$
  \ENDFOR
  \STATE \textbf{Output:} $\bar{u} = \frac{1}{T} \sum_{i = 1}^T u^t$
\end{algorithmic}
\end{algorithm}
\begin{theorem}[\cite{kakade2008}] \label{thm:smg}
Let $u^*$ be an optimal solution of $\min_{u \in \Delta(W,k)} c(u)$ for convex function $c$ on $\Delta(W,k)$. For SMG run on $c$, as long as $\E[ g^t|\mathcal{H}^{t-1}] = \nabla c(u^t)$ where $\mathcal{H}^{t-1}$ denotes history of previous iterations and $||\nabla c(u^t)||, || g^t||_\infty \leq B$ for all $ u$ on the simplex and for all iterations $t \le T$, then with probability $1- \delta$ for suitably chosen $\eta$
\[
c(\bar{u}) \leq c(u^*) + 4 BW \left(\sqrt{\frac{\log k}{T}} + \sqrt{\frac{2}{T}\log\frac{1}{\delta}}\right).
\]
\end{theorem}
To apply SMG, we convert our optimization problem's constraint from the $l_1$ ball to the simplex. To do so, we define the following mappings:
\[
\Pi_{B\rightarrow \Delta}^k: \mathsf{B}(W, k) \rightarrow \Delta(W, 2k+1), \quad \Pi_{B\rightarrow \Delta}^k(w)_{i} =
\begin{cases}
 W - ||w||_1 & \text{ if } i = 2k+1\\
w_i &\text{ if } i \in \{1, \ldots, k\}, w_i \ge 0\\
-w_{i- k} &\text{ if } i \in \{k+1, \ldots, 2k\}, w_{i - k} < 0\\
0 &\text{ otherwise.}
\end{cases}
\]
\[
\Pi_{\Delta \rightarrow B}^k: \Delta(W, 2k + 1) \rightarrow \mathsf{B}(W, k), \quad \Pi_{\Delta \rightarrow B}^k(u)_i = u_i - u_{i + k} \text{ for }i \in \{1, \ldots, k\}
\]
It is not hard to verify that these are valid mappings from ball to simplex and vice-versa. Using the given transformation, the optimization problem on the simplex is as follows:
\begin{align*}
\noalign{\textbf{Optimization Problem 2:}}
\min_{w \in \mathbb{R}^{2n-1}} &\quad\quad\widetilde{S}(w):= S(\Pi_{\Delta \rightarrow B}^{n-1}(w)) = \mathbb{E}_{Z \sim \D} \left[\exp\left(-\sum_{j=1}^{n-1}(w_{j} - w_{n- 1 + j}) Z_n Z_j\right) \right]\\
\text{subject to} &\quad\quad|| w||_1 = \lambda, w \ge 0.
\end{align*}

Note that the above loss is also convex and similar to Lemma \ref{lem:grad}, we can show the following, 
\begin{lemma} \label{lem:gradw}
Let $w^* = \Pi_{B\rightarrow \Delta}^{n-1}(v^*)$, then $\nabla \widetilde{S}(w^*) = 0$ and $w^*$ is a global minimum of Optimization Problem 2.
\end{lemma}
Using SMG we can thus solve Optimization Problem 2 as long as we have access to an unbiased estimator of the gradient.

%% file: Files/estimation.tex
\subsection{Unbiased Estimator of Gradient}
Let $v$ be a candidate solution of Optimization Problem 2.  Since we have missing/flipped data, it is not clear how to compute the loss or the gradient at the point $v$, and the gradient is required to execute a step of the SMG algorithm.  We will show, however, that it is possible to construct an unbiased estimator of the gradient with bounded $l_{\infty}$ norm for known $p_i$s. We then show that our estimates are sufficiently strong to plug-in to the SMG analysis. 

\subsubsection{Missing Data}
In the missing-data model, instead of getting samples from $\D$, we instead get samples from a corrupted distribution $\D_{\textsf{miss}}$ as follows:
\begin{itemize}
	\item Let $Z \sim \D$. Let $C_1 \sim \mathsf{Ber}(1 - p_1), \ldots, C_{n} \sim \mathsf{Ber}(1 - p_n)$.
	\item Output $X$ such that for all $i$, $X_i = C_iZ_i$ (replacing ? by 0).
\end{itemize} 
Here $\mathsf{Ber}(p)$ denotes the distribution of a Bernoulli random variable with probability $p$.

The main intuition of our estimator can be made clear with the following example. Suppose we wish to construct an unbiased estimator for function $f(Z_i)$ then consider $g(X_i) = \frac{f(X_i) - p f(0)}{1-p}$. Then it is not hard to see that,
\begin{align*}
&\E_{X_i}\left[\frac{f(X_i) - p f(0)}{1-p}\middle| Z_i\right]  = \underset{C_i \sim \mathsf{Ber}(1 - p_i)}{\E}\left[\frac{f(C_iZ_i) - p f(0)}{1-p}\middle| Z_i\right] \\
&= (1 - p) \times \frac{f(Z_i) - p f(0)}{1-p} + p \times f(0) = f(Z_i).
\end{align*}
Since ISO is a product function in $Z_1, \ldots, Z_{n-1}$, we can apply the above estimator to each term of the product. For $Z_n$ we use the above idea on the so formed product of estimators. This allows us to construct the following estimators, $\forall i \in [n-1]$,
 \[
g_{\textsf{miss}}^i(w;X) = -\frac{X_n}{1 - p_n} \times\frac{\exp\left(-(w_{i} - w_{n- 1 + i})X_nX_i\right)X_i}{1 - p_i} \times\prod_{j \neq i,n}\frac{\exp\left(-(w_{j} - w_{n- 1 + j})X_nX_j\right) - p_j}{1- p_j}.
\]

The following lemma shows how to use $g_{\textsf{miss}}$ to construct an unbiased estimator of the gradient of $\widetilde{S}$.
\begin{lemma} \label{lem:miss}
Consider estimator $G_{\textsf{miss}}(w;X) = \sum_{i=1}^{n-1} g_{\textsf{miss}}^i(w;X) (e^i - e^{n - 1 + i})$ where $e^i \in \mathbb{R}^{2n-1}$ is the indicator vector for coordinate $i$. Then for all fixed $w$,
\[
\underset{X\sim \D_{\textsf{miss}}}{\E}[G_{\textsf{miss}}(w;X)] = \nabla \widetilde{S}(w).
\]
Also, for all $X \in \{0,1\}^n$ and $w \in \Delta{B}(\lambda, 2n-1)$,
\[
||G_{\textsf{miss}}(w;X)||_{\infty} \le \frac{1}{(1-p_{\max})^{2}} \exp\left(\frac{\lambda}{1-p_{\max}}\right)
\]
where $p_{\max} = \max_i p_i$.
\end{lemma}
\begin{proof}
Observe that the gradient of $\widetilde{S}(w)$ with respect to $w$ can be computed as follows. We have $\nabla \widetilde{S}(w)_{2n-1} = 0$ since $\widetilde{S}$ does not depend on $w_{2n-1}$, and
\[
\forall i \in [n-1],
\nabla \widetilde{S}(w)_i = -\nabla \widetilde{S}(w)_{i + n - 1} =  -\mathbb{E}_{Z\sim \D} \left[\exp\left(-\sum_{j=1}^{n-1}(w_{j} - w_{n- 1 + j}) Z_n Z_j\right) Z_n Z_i\right].
\]
For ease of presentation, let $v = \Pi^{n-1}_{\Delta \rightarrow B}(w)$ that is, $v_i = w_i - w_{n-1+i}$ for $i \in [n-1]$. Taking expectation of $g_{\textsf{miss}}$ over $\D_{\mathsf{miss}}$, we have
\begin{align*}
\underset{X \sim \D_{\textsf{miss}}}{\E}[g_{\textsf{miss}}^i(w;X)] &= -\underset{X\sim \D_{\textsf{miss}}}{\E}\left[\frac{X_n}{1- p_n} \times\frac{\exp\left(-v_iX_nX_i\right)X_i}{1-p_i)} \times\prod_{j \neq i,n}\frac{\exp\left(-v_jX_nX_j\right) - p_i}{1- p_i}\right]\\
&= -\underset{Z \sim \D}{\E}\left[\frac{C_nZ_n}{1 - p_n} \times \underset{C_i \sim \mathsf{Ber}(1 - p_i)}{\E} \left[\frac{\exp\left(-v_iC_nZ_nC_iZ_i\right)C_iZ_i}{1 - p_i}\right] \times\right.\\
& \qquad \qquad \left. \prod_{j \neq i,n}\underset{C_j \sim \mathsf{Ber}(1 - p_j)}{\E}\left[\frac{\exp\left(-v_jC_nZ_nC_jZ_j\right) - p_j}{1- p_j}\right]\right]\\
&= -\underset{\substack{Z \sim \D\\ C_n \sim \mathsf{Ber}(1- p_n)}}{\E}\left[\frac{C_nX_n}{1- p_n}\times \exp\left(-v_iC_nZ_nZ_i\right)Z_i \times \prod_{j \neq i,n}\exp\left(-v_jC_nZ_nZ_j\right)\right]\\
&= - \underset{\substack{Z \sim \D\\ C_n \sim \mathsf{Ber}(1- p_n)}}{\E}\left[\frac{\exp\left(-\sum_{j\neq n}v_jC_nZ_nZ_j\right)C_nZ_nZ_i}{1-p_n}\right]\\
&= -\underset{Z \sim \D}{\E}\left[\exp\left(-\sum_{j \neq n}v_jZ_nZ_j\right)Z_n Z_i\right] = \nabla \widetilde{S}(w)_i.
\end{align*}
The above follows from observing the following facts: 1) $\E_{C \sim \mathsf{Ber}(1- p)}\left[\frac{\exp\left(C X\right) - p}{1- p}\right] = \exp\left(X\right)$ as long as $X$ is independent of $C$, 2) $\E_{C \sim \mathsf{Ber}(1- p)}\left[\frac{\exp\left(C X\right)CY}{1- p}\right] = \exp\left(X\right)Y$ as long as $X,Y$ are independent of $C$.

The last thing we need is that the above estimator has bounded $l_\infty$ norm for all $ w \in \Delta(\lambda, 2n-1)$ ($v \in B(\lambda, n-1)$). Observe that
\[
\left|\frac{\exp(a) - p}{1- p}\right| = \frac{\left|1-p + \sum_{i = 1}^\infty \frac{a^i}{i!}\right|}{1 - p} \leq 1 + \sum_{i = 1}^\infty \frac{|a|^i}{(1- p)i!} \leq 1 + \sum_{i = 1}^\infty \frac{\left(\frac{a}{1-p}\right)^i}{i!} = \exp\left(\frac{|a|}{1-p}\right).
\]
Using the above property,
\begin{align*}
|g_{\textsf{miss}}^i(w;X)| &\leq \left|\frac{X_n}{1- p_n}\right| \times\left|\frac{\exp\left(-v_iX_nX_i\right)X_i}{1- p_i}\right| \times \prod_{j \neq i,n}\left|\frac{\exp\left(-v_jX_nX_j\right) - p_i}{1- p_i}\right| \\
& \frac{\exp\left(|v_i|\right)}{(1- p_{\max})(1-p_{\max})} \prod_{j \neq i,n}\exp\left(\frac{|v_i|}{1 - p_{\max}}\right)\\
& \leq \frac{1}{(1-p_{\max})^{2}} \exp\left(\frac{\lambda}{1-p_{\max}}\right).
\end{align*}
Here the inequality follows from observing that $\sum_{i=1}^{n-1}|v_i| = ||v||_1 \le \lambda$. Thus $||G_{\textsf{miss}}(w;X)||_{\infty} \leq \frac{1}{(1-p_{\max})^{2}} \exp\left(\frac{\lambda}{1-p_{\max}}\right)$ for all $w\in \Delta(\lambda, 2n-1)$.
\end{proof}

\subsubsection{Random-Flipped Data}
In the random-flipped data model, instead of getting samples from $\D$, we instead get samples from a corrupted distribution $\D_{\textsf{flip}}$ as follows:
\begin{itemize}
	\item Let $Z \sim \D$. Let $C_1 \sim \mathsf{Rad}(1 - p_1), \ldots, C_{n} \sim \mathsf{Rad}(1 - p_n)$.
	\item Output $X$ such that for all $i$, $X_i = C_iZ_i$.
\end{itemize} 
Here $\mathsf{Rad}(p)$ denotes the distribution of Rademacher variables with probability $p$, that is, distribution over $\{-1,1\}$ where probability of drawing 1 is $p$.

Similar to the missing data case, we motivate our estimator with the following example. Suppose we wish to construct an unbiased estimator for any function $f(Z_i)$ then consider $g(X_i) = \frac{(1 - p)f(X_i) - p f(-X_i)}{1-2p}$. Then it is not hard to see that,
\begin{align*}
&\E_{X_i}\left[\frac{(1 - p)f(X_i) - p f(-X_i)}{1-2p}\middle| Z_i\right] \\
& = \underset{C_i \sim \mathsf{Rad}(1 - p_i)}{\E}\left[\frac{(1 - p)f(C_iZ_i) - p f(-C_iZ_i)}{1-2p}\middle| Z_i\right] \\
&= (1 - p)\times \frac{(1 - p)f(Z_i) - p f(-Z_i)}{1-2p} + p \times \frac{(1 - p)f(-Z_i) - p f(Z_i)}{1-2p} = f(Z_i).
\end{align*}
Based on the above example, define $\sigma(p,a,x) = \frac{(1- p)\exp(ax) -p\exp(-ax)}{1 - 2p}$. Consider the following estimators, for all $i \in [n-1]$:
\begin{align*}
& g_{\textsf{flip}}^i(w;X) = - \frac{(1-p_n) h^i(w;X) + p_n h^i(-w;X)}{1 - 2p_n}\\
\text{where } &h^i(w;X) = X_n X_i \times \prod_{j \neq n}\sigma\left(p_j, -(w_{j} - w_{n- 1 + j})X_n,X_j\right).
\end{align*}
The following lemma shows that $g_{\textsf{flip}}$ is indeed an unbiased estimator of $\nabla\bar{S}(w)_i$.
\begin{lemma} \label{lem:flip}
Using the above, we construct estimator $G_{\textsf{flip}}(w;X) = \sum_{i=1}^{n-1} g_{\textsf{flip}}^i(w;X) (e^i - e^{n - 1 + i})$ where $e^i \in \mathbb{R}^{2n-1}$ is the indicator vector for coordinate $i$. Then for all fixed $w$,
\[
\underset{X\sim \D_{\textsf{flip}}}{\E}[G_{\textsf{flip}}(w;X)] = \nabla \widetilde{S}(w).
\]
Also, for all $X \in \{-1,1\}^n$ and $w \in \Delta{B}(\lambda, 2n-1)$,
\[
||G_{\textsf{flip}}(w;X)||_{\infty} \le \frac{1}{(1-p_{\max})^{2}} \exp\left(\frac{\lambda}{1-p_{\max}}\right)
\]
where $p_{\max} = \max_i p_i$.
\end{lemma}
We defer the proof to the appendix as it follows roughly the same ideas as of the missing-data model.

%% file: Files/maintheorem.tex
\begin{algorithm}[t]\label{alg:smgm}
  \caption{SMG with Missing/Flipped Data on Ising Models}
\begin{algorithmic}
  \STATE Set $w^1_i := \lambda / (2n - 1)$ for all $i$.
  \IF {Missing Data}
  \STATE $ G(w;X):= G_{\textsf{miss}}(w;X)$
 \ENDIF
 \IF {Flipped Data}
  \STATE $ G(w;X):= G_{\textsf{flip}}(w;X)$
 \ENDIF
  \FOR{$t=1$ {\bfseries to} $T$}
  \STATE Draw missing/flipped data sample $X^t$ from $\D_{\mathsf{miss}}$/$\D_{\mathsf{flip}}$ respectively.
  \STATE Let $C^t := \frac{(w^t)^T g(w^t;X^t)}{\lambda}$
  \STATE For all $i$, update $w_i^{t+1} := w_i^t\left(1 - \eta G(w^t;X^t)_i + \eta C^t \right)$
  \ENDFOR
  \STATE {Output $\bar{w} = \frac{1}{T} \sum_{i = 1}^T w^t$}
\end{algorithmic}
\end{algorithm}
Combining the techniques presented in the previous sections, our main algorithm (Algorithm \ref{alg:smgm}) gives us the following guarantees.
\begin{theorem}\label{thm:main}
For known constants $p_i$, given samples with missing data from an unknown $n$-variable Ising model $\D(A, 0)$, for $T \ge \lambda^4\log (n/\delta) \exp(C\lambda)\epsilon^{-4}$ with large enough constant $C>0$, Algorithm \ref{alg:smgm} returns $\bar{w}$ such that:
\[
\forall~ i \in [n-1],~\left|\Pi_{\Delta \rightarrow B}^{n-1}(\bar{w})_i - A_{ni}\right| \leq \epsilon.
\]
\end{theorem}
\begin{proof}
Observe that $X^t$ is a new draw and is independent of $X^1, \ldots, X^{t-1}$ and therefore $w^t$, thus $G(w^t;X^t)$ is an unbiased estimator of $\nabla \widetilde{S}(w^t)$ (using Lemma \ref{lem:miss} and \ref{lem:flip}) conditioned on $X^1, \ldots, X^{t-1}$. Applying Theorem \ref{thm:smg} gives us that the output of Algorithm \ref{alg:smgm}, $\bar{w}$, satisfies
\[
\widetilde{S}(\bar{w}) - \widetilde{S}(w^*) \le \frac{4 \lambda}{(1-p_{\max})^2} \exp\left(\frac{\lambda}{1-p_{\max}}\right) \left(\sqrt{\frac{\log n}{T}} + \sqrt{\frac{2}{T}\log \frac{1}{\delta}}\right).
\]
 Recall that $w^* = \Pi^{n-1}_{B \rightarrow \Delta}(v^*)$ where $v^*$ satisfies $v^*_i = A_{ni}$ for $i \in [n-1]$ such that $(i,n) \in E$ else 0. By definition of the mappings, we can see that $v^* = \Pi^{n-1}_{\Delta \rightarrow B}(w^*)$ and $\widetilde{S}(w^*) = S(v^*)$. Using Lemma \ref{lem:grad}, we have $\nabla S(v^*) = 0$, thus choosing $\bar{v} = \Pi^{n-1}_{\Delta \rightarrow B}(\bar{w})$ and applying the RSC property of $S$ (see Lemma \ref{lem:rsc}),
\begin{align*}
 ||\bar{v} - v^*||_\infty^2 &\le (1 + \lambda)\exp(3\lambda) (S(\bar{v})) - S(v^*))\\
 & = (1 + \lambda)\exp(3\lambda) (\widetilde{S}(\bar{w}) - \widetilde{S}(w^*))\\
 &\le \frac{4(1 + \lambda) \lambda}{(1-p_{\max})^2} \exp\left(\lambda\left(\frac{1}{1-p_{\max}} + 3\right)\right) \left(\sqrt{\frac{\log n}{T}} + \sqrt{\frac{2}{T}\log \frac{1}{\delta}}\right).
\end{align*}
Thus for given $T = O\left(\frac{\lambda^4}{(1-p_{\max})^4} \exp\left(\lambda\left(\frac{2}{1-p_{\max}} + 6\right)\right) \log (n/\delta)\right)$ we have
\[
\forall~ i\in [n-1],~|\bar{v}_i - v^*_i|=|\bar{v}_i - A_{ni}| \leq \epsilon.
\]
\end{proof}
Assuming $p_{\max}$ as a constant gives us the required result.

Similarly for flipped data, we get:
\begin{theorem}
For known constants $p_i < 0.5$, given samples with flipped data from an unknown $n$-variable Ising model $\D(A, 0)$, for $T \ge \lambda^4\log (n/\delta) \exp(C\lambda)\epsilon^{-4}$ with large enough constant $C>0$, Algorithm \ref{alg:smgm} returns $\bar{w}$ such that:
\[
\forall~ i \in [n-1],~\left|\Pi^{n-1}_{\Delta \rightarrow B}(\bar{w})_i - A_{ni}\right| \leq \epsilon.
\]
\end{theorem}

Setting $\epsilon = \beta/2$ in the above theorems where $\beta$ is the smallest edge weight will recover the neighborhood of vertex $n$ exactly. To recover the entire graph, we can run Algorithm \ref{alg:smgm} for each vertex using the same samples. Thus with probability $1 - \delta$, using $T = O\left(\frac{\lambda^4}{\beta^4}\log (n^2/\delta) \exp(C\lambda)\right)$ samples we will recover the entire graph. The runtime for Algorithm \ref{alg:smgm} is $O(nT)$ as the unbiased estimator requires $O(n)$ time to compute. Thus, the overall runtime to recover the entire graph is $O(n^2T) = \widetilde{O}(n^2)$. 

%% file: Files/unknownp.tex
In this section, we will show how to extend the analysis for the missing data model to the case when $p_i = p$ for all $i$ and $p$ is unknown. The main approach is to use fresh samples to estimate $p$ with $\widehat{p}$ such that $p - \widehat{p} \approx \delta/\sqrt{Tn}$ using $T/\delta^2$ samples where $T$ corresponds to the number of iterations needed for the SMG algorithm. We can compute an empirical estimate of $p$ since we observe $?$ when an entry is missing (it is unclear how to do this for the flipped data model). Also note that there is no dependence on $n$ since each sample gives $n$ estimates for $p$, one for each coordinate. Subsequently we will show that the distribution of $T$ samples using $p$ and $\widehat{p}$ are within total variation distance $O(\delta)$ for constant $p$.  It follows that we can use $\widehat{p}$ in our SMG algorithm and obtain the same guarantees while losing a factor of $O(\delta)$ in the failure probability.

\subsection{Estimating \boldmath$p$}
Prior to running SMG, we will draw $m$ samples. Let $\widehat{p}$ be the fraction of $?$ observed in the $m$ samples. Since each $?$ is i.i.d., with probability $p$, we have by Chernoff,
\[
\Pr[|p - \widehat{p}| \ge \epsilon] \le \exp\left(\frac{-mn \epsilon^2}{2p(1-p)}\right)
\]
For $\epsilon = \frac{\delta}{\sqrt{Tn}}$, we have with probability $1 - \delta$, $|p - \widehat{p}| \le \frac{\delta}{\sqrt{Tn}}$ using $m = \widetilde{O}(T/\delta^2)$ samples.

\subsection{Distribution closeness using \boldmath$\hat{p}$}
Here we will show that the distribution $\mathcal{D}$ over $T$ samples with missing probability $p$ is $\delta$ close to the distribution $\widehat{\mathcal{D}}$ over $T$ samples with missing probability $\widehat{p}$. In order to do so, we will first show that $\mathcal{D}$ is equivalent to the following distribution $\mathcal{D}_{\textsf{perm}}$:
\begin{enumerate}
  \item Concatenate all $T$ samples into a $nT$-dimensional vector and permute all coordinates.
  \item Select first $c \sim \mathsf{Bin}(nT, p)$ coordinates and replace with $?$.
  \item Unpermute and split back to $T$ samples.
\end{enumerate}
Here $\mathsf{Bin}(n,p)$ denotes the binomial distribution over $n$ runs with success probability $p$. It is not hard to see that $\mathcal{D}_{\textsf{perm}}$ is equivalent to $\mathcal{D}$. Similarly the distribution $\widehat{\mathcal{D}}_{\mathsf{perm}}$, analogous to $\mathcal{D}_{\textsf{perm}}$ with $p$ replaced by $\hat{p}$ is equivalent to $\widehat{\mathcal{D}}$. Thus, we have
\[
d_{\mathsf{TV}}(\mathcal{D}, \widehat{\mathcal{D}}) = d_{\mathsf{TV}}(\mathcal{D}_{\mathsf{perm}}, \widehat{\mathcal{D}}_{\mathsf{perm}}) = d_{\mathsf{TV}}(\mathsf{Bin}(nT, p), \mathsf{Bin}(nT, \hat{p})).
\]
The above follows as $\mathcal{D}_{\mathsf{perm}}$ and $\widehat{\mathcal{D}}_{\mathsf{perm}}$ differ only in Step 2, which depends on only closeness of the binomial distribution.

We will use the following lemma to bound the closeness of binomial distributions.
\begin{lemma}[\cite{roos2001binomial}]
For $0< p < 1$ and $0 < \delta < 1 - p$,
\[
d_{\mathsf{TV}}(\mathsf{Bin}(n, p), \mathsf{Bin}(n, p + \delta)) \le \frac{\sqrt{e}}{2}\frac{\theta(\delta)}{(1 - \theta(\delta))^2}\quad\text{where}\quad \theta(\delta):= \delta \sqrt{\frac{n+2}{2p(1-p)}}.
\]
\end{lemma}
The above gives us,
\[
d_{\mathsf{TV}}(\mathcal{D},\widehat{\mathcal{D}}) = d_{\mathsf{TV}}(\mathsf{Bin}(nT, p), \mathsf{Bin}(nT, \hat{p})) \le O(\delta).
\]
This implies that with probability $1 - O(\delta)$, the draw from $\mathcal{D}$ will be identical to that of $\widehat{\mathcal{D}}$.  Thus, Algorithm \ref{alg:smgm} would give the desired result with estimate $\hat{p}$.

%% file: Files/appendix.tex
\section{Omitted Proofs}
\subsection{Proof of Lemma \ref{lem:grad}}
Computing the gradient, we have
\begin{align*}
\nabla S(v^*)_i &= - \mathbb{E}_{Z \sim \D} \left[\exp\left(-\sum_{j| (n,j) \in E}A_{nj} Z_n Z_j\right) Z_n Z_i\right]\\
& = -\frac{1}{\mathcal{Z \sim \D}} \sum_{z \in \{-1, 1\}^n} \exp\left(-\sum_{j| (n,j) \in E}A_{nj} z_n z_j \right) \exp\left(\sum_{(i,j) \in E}A_{ij} z_iz_j\right)z_n z_i\\
& = -\frac{1}{\mathcal{Z \sim \D}} \left(\sum_{z_{-n} \in \{-1, 1\}^{n-1}} \exp\left(\sum_{i, j \neq n, (i,j) \in E}A_{ij} z_iz_j\right) z_i\right) \left(\sum_{z_n \in \{-1, 1\}}z_n\right) = 0
\end{align*}
Since $S$ is convex, gradient is 0 at $v^*$ and $v^*$ lies on the $l_1$-ball of radius $\lambda$ by assumption, it is the global minimum.

\subsection{Proof of Lemma \ref{lem:rsc}}
To prove the above property, we use Lemma 5 from \cite{vuffray2016interaction} to get for all $ v$,
\[
S(v) - S(v^*) \ge \nabla S(v^*)^T(v - v^*) + \frac{\exp(-\lambda)}{1 + || v - v^*||_1}(v - v^*)^T H (v - v^*)
\]
where $ H $ satisfies $H_{ij} = \E_{Z \sim \D}[Z_i Z_j]$ for $i,j \in [n-1]$. The above property in their paper was proved for the empirical loss, however the same proof goes through for the expected loss with $H$ being the true covariance matrix instead of the empirical one.

Let $\Delta = v- v^*$, now to bound $ \Delta^T H \Delta$, we give an alternate stronger bound which holds for any vector on the $l_1$-ball and does not require some implicit sparsity. 
Note that we have $\Delta^T H \Delta = \sum_{i,j = 1}^{n-1} \Delta_i \Delta_j \E_{Z \sim \D}[Z_i Z_j] = \E_{Z \sim \D}\left[\left(\sum_{i=1}^{n-1}\Delta_i Z_i\right)^2\right]$. Now for any $j \neq n$, we can condition the expectation as follows:
\begin{align*}
 &\Delta^T H \Delta \\
 &= \sum_{z_{-n} \in \{-1,1\}^{n-1}}\Pr[Z_{-n} = z_{-n}]\left(\sum_{i=1}^{n-1}\Delta_i Z_i\right)^2\\
&= \sum_{z_{-\{j,n\}} \in \{-1,1\}^{n-2}}\Pr[Z_j = 1| Z_{-\{j,n\}} = z_{-\{j,n\}}]\Pr[Z_{-\{j,n\}} = z_{-\{j,n\}}]\left(\Delta_j + \sum_{i=1, i \neq j}^{n-1}\Delta_i z_i\right)^2 \\
&\quad + \sum_{z_{-\{j,n\}}\in \{-1,1\}^{n-2}} \Pr[Z_j = -1| Z_{-\{j,n\}} = z_{-\{j,n\}}]\Pr[Z_{-\{j,n\}} = z_{-\{j,n\}}]\left(-\Delta_j + \sum_{i=1, i \neq j}^{n-1}\Delta_i z_i\right)^2\\
&\geq \frac{\exp(-2\lambda)}{2}\sum_{z_{-\{j,n\}}\in \{-1,1\}^{n-2}}\Pr[Z_{-\{j,n\}} = z_{-\{j,n\}}]\left(\left(\Delta_j + \sum_{i=1, i \neq j}^{n-1}\Delta_i z_i\right)^2 + \left(-\Delta_j + \sum_{i=1, i \neq j}^{n-1}\Delta_i z_i\right)^2\right)\\
&\geq \frac{\exp(-2\lambda)}{2}\sum_{z_{-\{j,n\}}\in \{-1,1\}^{n-2}}\Pr[Z_{-\{j,n\}} = z_{-\{j,n\}}]\left(2\Delta_j^2 + 2\left(\sum_{i=1, i \neq j}^{n-1}\Delta_i z_i\right)^2 \right)\\
&\geq \exp(-2\lambda) v_j^2 \sum_{z_{-\{j,n\}}\in \{-1,1\}^{n-2}}\Pr[Z_{-\{j,n\}} = z_{-\{j,n\}}] = \exp(-2\lambda) \Delta_j^2.
\end{align*}
In the above, the first inequality follows from observing that $\Pr[Z_j = 1|Z_{-\{j,n\}} = z_{-\{j,n\}}] \ge \frac{\exp(-2\lambda)}{2}$ using Lemma \ref{lem:bresler}. Thus we get, $ \Delta^T H \Delta \geq \exp(-2\lambda)||\Delta||_{\infty}^2$ which gives us the result.

\subsection{Proof of Theorem \ref{thm:smg}}
Since $c$ is convex, we have for all iterations $c(u^*) - c(u^t) \ge \nabla c(u^t)^T (u^* - u^t)$. Applying Jensen's inequality, we have,
\[
c(\bar{u}) - c(u^*) \leq \frac{1}{T} \sum_{i = 1}^T (c(u^t) - c(u^*)) \leq \frac{1}{T} \sum_{i = 1}^T \nabla c(u^t)^T (u^t - u^*).
\]
Since $\E[ g^t|\mathcal{H}^{t-1}] = \nabla c(u^t)$, we have $M^t = \sum_{r = 1}^t(g^r - \nabla c(u^r))^T (u^r - u^*)$ is a martingale with respect to the filtration $\mathcal{H}^{t-1}$. Also $\E[M^1] = 0$ and $|M^t - M^{t-1}| = (g^t - \nabla c(u^t))^T (u^t - u^*) \leq 4 BW$. By Azuma Hoeffding bound, we have with probability $1- \delta$,
\[
\frac{1}{T} \sum_{i = 1}^T (\nabla c(u^t) - g^t)^T (u^t - u^*) \leq 4 BW\sqrt{\frac{2}{T}\log\frac{1}{\delta}}.
\]
Using the above, we have
\[
c(\bar{u}) - c(u^*) \le \frac{1}{T} \sum_{i = 1}^T (g^t)^T (u^t - u^*) + 4 BW\sqrt{\frac{2}{T}\log\frac{1}{\delta}}.
\]
Let us now bound the first term. Note that for $\eta \leq \frac{1}{8B}$ we have $|\eta g_i^t - \eta C^t| \leq 2\eta B \leq 1/4$. Thus, each update keeps the weights positive. Also observe that, $|| u^{t+1}||_1 = || u^{t}||_1$ by construction, implying that each $ u^t$ remains in the feasible set.

For simplicity, assume $W = 1$. Then all feasible solutions denote a probability distribution. Consider the following:
\begin{align*}
\mathsf{KL}(u^*|| u^t) - \mathsf{KL}(u^*|| u^{t+1}) &= \sum_{i=1}^n u^*_i \log \frac{u^{t+1}_i}{u^t_i}\\
&= \sum_{i=1}^n u^*_i \log (1 - \eta g^t_i + \eta (g^t)^T u^t)\\
&\ge \sum_{i=1}^n u^*_i(- \eta g^t_i + \eta (g^t)^T u^t) - \sum_{i=1}^n u^*_i(- \eta g^t_i + \eta (g^t)^T u^t)^2\\
&\ge \eta (g^t)^T(u^t - u^*) - 4 \eta^2 B^2.
\end{align*}
Here the first inequality follows from observing that $\log(1 + x) \geq x - x^2$ for all $x$ such that $x \ge - 1/4$. Rearranging, we have,
\begin{align*}
\frac{1}{T}\sum_{i = 1}^T (g^t)^T (u^t - u^*) &\le \frac{1}{\eta T}(\mathsf{KL}(u^*|| u^1) - \mathsf{KL}(u^*|| u^{t+1})) + 4 \eta B^2\\
&\le \frac{\log k}{\eta T} + 4 \eta B^2\\
&\le \sqrt{\frac{\log n}{T}}
\end{align*}
for $\eta \le \frac{1}{2 B}\sqrt{\frac{\log k}{T}}$. The first inequality follows from observing that $\mathsf{KL}(u^*|| u^1) \le \log k$. We can now rescale by $W$ to get the required result.

\subsection{Proof of Lemma \ref{lem:flip}}
Similar to the proof of \ref{lem:miss}, let $v = \Pi^{n-1}_{\Delta \rightarrow B}$ that is, $v_i = w_i - w_{n-1+i}$. Taking expectation of $g_{\textsf{miss}}$ over $\D_{\mathsf{flip}}$, we have
\begin{align*}
\underset{X \sim \D_{\mathsf{flip}}}{\E}[f(w;X)_i] &= \underset{X\sim \D_{\mathsf{flip}}}{\E}\left[X_n X_i\times \prod_{j \neq n}\sigma(p_j, -v_jX_n,X_j)\right]\\
&= \underset{\substack{Z \sim \D\\ C_n \sim \mathsf{Rad}(1- p_n)}}{\E}\left[C_nZ_n\underset{C_i \sim \mathsf{Rad}(1 - p_i)}{\E}\left[C_iZ_i\sigma(p_i, -v_iC_nZ_n, C_iZ_i)\right] \times\right.\\
&\qquad \qquad \left. \prod_{j \neq i,n}\underset{C_j \sim \mathsf{Rad}(1 - p_j)}{\E}\left[\sigma(p_j, -v_jC_nZ_n,C_j Z_j)\right]\right]\\
&= \underset{\substack{Z \sim \D\\ C_n \sim \mathsf{Rad}(1- p_n)}}{\E}\left[C_n Z_n \times \exp\left(-v_iC_nZ_nZ_i\right)Z_i\times \prod_{j \neq i,n}\exp\left(-v_jC_nZ_nZ_j\right)\right]\\
&= \underset{\substack{Z \sim \D\\ C_n \sim \mathsf{Rad}(1- p_n)}}{\E}\left[\exp\left(-\sum_{j=1}^{n-1}v_iX_nZ_j\right)C_nZ_nZ_i\right]\\
&= \underset{Z \sim \D}{\E}\left[\left((1-p_n)\exp\left(-\sum_{j=1}^{n-1}v_iZ_nZ_j\right) - p_n\exp\left(\sum_{j=1}^{n-1}v_iZ_nZ_j\right)\right)Z_nZ_i\right].
\end{align*}
The above follows from observing the following two facts that follow from the example mentioned in the discussion: for all $i$ 1) $\E_{C_i}[\sigma(p, A, C_iZ_i)] = \exp(AZ_i)$, 2) $\E_{C_i}[X_i\sigma(p, A, C_iZ_i)] = \exp(AZ_i)Z_i$ as long as $A$ is independent of $C_i$. 

Using the above, we have
\begin{align*}
 \underset{X \sim \D_{\mathsf{flip}}}{\E}[g_{\textsf{flip}}^i(w;X)] &= -\frac{(1-p_n)\E_{X}[f(w;X)_i] + p_n \E_{X}[f(- w;X)_i]}{1 - 2p_n}\\
 &= -\underset{Z \sim \D}{\E}\left[\exp\left(-\sum_{j=1}^{n-1}v_iZ_nZ_j\right)Z_nZ_i\right] = \nabla \widetilde{S}(w)_i.
\end{align*}
The last thing we need is that the above estimator has bounded $l_\infty$ norm for all $ w \in \Delta(\lambda, 2n-1)$ ($v \in B(\lambda, n-1)$). We have
\begin{align*}
|\sigma(p, a,X)| &= \frac{\left|(1- p)\exp(aX) -p\exp(-aX)\right|}{1 - 2p} \\
&= \frac{\left|(1- p)\left(1 + \sum_{i = 1}^\infty \frac{(aX)^i}{i!}\right) -p\left(1 + \sum_{i = 1}^\infty (-1)^i\frac{(aX)^i}{i!}\right)\right|}{1 - 2p}\\
&= \left|1 + \sum_{i = 1}^\infty \left(\frac{(aX)^{2i}}{(2i)!} + \frac{(aX)^{2i - 1}}{(1-2p)(2i-1)!}\right)\right|\\
&\le 1 + \sum_{i = 1}^\infty\frac{|aX|^{i}}{(1-2p)i!} \le \exp\left(\frac{|aX|}{1 - 2p}\right).
\end{align*}
Let $p_{\max}:= \max_i p_i$, then using the fact that $|X| = 1$ we get
\begin{align*}
|f(w; X)^i| &\le \frac{\exp(|v_i|)}{1 - 2p_i} \prod_{j \neq i,n} \exp\left(\frac{|w_j - w_{n-1+j}|}{1-2p_j}\right) \\
&= \frac{1}{1-2p_{\max}}\exp\left(\frac{\lambda}{1 - 2p_{\max}}\right).
\end{align*}
Here the inequality follows from observing that $\sum_{i=1}^{n-1}|v_i| = ||v||_1 \le \lambda$. Using above, we get the final bound on $||G_{\textsf{flip}}(w;X)||_{\infty} \leq \frac{1}{(1-p_{\max})^{2}} \exp\left(\frac{\lambda}{1-p_{\max}}\right)$.

\section{Ising Models with Non-zero Mean Field} \label{app:nonzero}

In this section, we will extend the analysis to the setting $\theta \ne 0$. Similar to the case of zero-mean field, consider the neighborhood of a fixed vertex (WLOG we choose $n$, we can do the same analysis for any $v \in V$). We will work with the following modified ISO:
\begin{align*}
\min_{v \in \mathbb{R}^{n}} & \quad\quad \underline{S}(v) := \mathbb{E}_{Z \sim \D} \left[\exp\left(-\sum_{j=1}^{n-1}v_{j} Z_n Z_j - v_nZ_n\right) \right]\\
\text{subject to} &\quad\quad || v||_1 \leq \lambda_{\ne 0}.
\end{align*}
Define $\underline{v}^* \in \mathbb{R}^{n-1}$ such that $\underline{v}^*_n = \theta_n$ and for $i \in [n-1]$, $\underline{v}^*_i = A_{ni}$ if $(i,n) \in E$  else 0.
It is not hard to see that the above loss is convex and following the same line of proof of Lemma \ref{lem:grad}, one can also show that $\nabla \underline{S}(\underline{v}^*) = 0$ making it an optimal solution. Further, one can prove a close to RSC property for the given objective.
\begin{lemma}
For all $ v$,
\[
\underline{S}(v) - \underline{S}(\underline{v}^*) \ge \nabla \underline{S}(\underline{v}^*)^T(v - \underline{v}^*) + \frac{\exp(-3\underline{\lambda})}{1 + \lambda}||(v - \underline{v}^*)_{-n}||_\infty.
\]
\end{lemma}
\begin{proof}
Consider the following transformation for the input variables of the Ising model, $Z \rightarrow (Z,1)$.  We can then use Lemma 5 from \cite{vuffray2016interaction} to get for all $ v$,
\[
\underline{S}(v) - \underline{S}(v^*) \ge \nabla \underline{S}(v^*)^T(v - v^*) + \frac{\exp(-\lambda)}{1 + || v - v^*||_1}(v - v^*)^T \underline{H} (v - v^*)
\]
where $\underline{H}$ is as follows:
\[
\underline{H}_{ij} = 
\begin{cases}
1 &\text{ for } i = j = n\\
\E_{Z \sim \D}[Z_i] &\text{ for } i \in [n-1], j = n\\
\E_{Z \sim \D}[Z_j] &\text{ for } i = n, j \in [n-1]\\
\E_{Z \sim \D}[Z_i Z_j] &\text{ otherwise.}\\
\end{cases}
\]
We will now bound $v^T\underline{H}v$, we have, $v^T H v = \E_{Z \sim \D}\left[\left(\sum_{i=1}^{n-1}v_i Z_i + v_n\right)^2\right]$. Similar to the proof of Lemma \ref{lem:rsc}, for all $j \in [n-1]$,
\begin{align*}
&v^T \underline{H} v \\
&= \sum_{z_{-n} \in \{-1,1\}^{n-1}}\Pr[Z_{-n} = z_{-n}]\left(\sum_{i=1}^{n-1}v_i Z_i + v_n\right)^2\\
&= \sum_{z_{-\{j,n\}} \in \{-1,1\}^{n-2}}\Pr[Z_j = 1| Z_{-\{j,n\}} = z_{-\{j,n\}}]\Pr[Z_{-\{j,n\}} = z_{-\{j,n\}}]\left(v_j + \sum_{i=1, i \neq j}^{n-1}v_i z_i + v_n\right)^2 \\
&\quad + \sum_{z_{-\{j,n\}}\in \{-1,1\}^{n-2}} \Pr[Z_j = -1| Z_{-\{j,n\}} = z_{-\{j,n\}}]\Pr[Z_{-\{j,n\}} = z_{-\{j,n\}}]\left(-v_j + \sum_{i=1, i \neq j}^{n-1}v_i z_i + v_n\right)^2\\
&\geq \frac{\exp(-2\underline{\lambda})}{2}\sum_{z_{-\{j,n\}}\in \{-1,1\}^{n-2}}\Pr[Z_{-\{j,n\}} = z_{-\{j,n\}}]\left[ 2v_j^2 + 2\left(\sum_{i=1, i \neq j}^{n-1}v_i z_i + v_n\right)^2 \right]\\
&\geq \exp(-2\underline{\lambda}) v_j^2.
\end{align*}
Thus, we have that $v^T\underline{H}v \ge \exp(-2\underline{\lambda})||v_{-n}||^2_{\infty}$. Substituting this in the previous equation gives us the desired result.
\end{proof}
Similar to the zero mean-field case, we can transform the problem from $B(\lambda, n)$ to $\Delta(\lambda, 2n + 1)$ to get the following optimization problem:
\begin{align*}
\min_{w \in \mathbb{R}^{2n+1}} &\quad\quad\widetilde{\underline{S}}(w):= \underline{S}(\Pi^{n}_{\Delta \rightarrow B}(w)) = \mathbb{E}_{Z \sim \D} \left[\exp\left(-\sum_{j=1}^{n-1}(w_{j} - w_{n + j}) Z_n Z_j - (w_n - w_{2n})Z_n\right) \right]\\
\text{subject to} &\quad\quad|| w||_1 = \lambda, w \ge 0.
\end{align*}
Note that the above loss is also convex and similar to Lemma \ref{lem:grad}, we can show that $\underline{w}^* = \Pi_{B\rightarrow \Delta}^{n}(\underline{v}^*)$ for $\underline{v}^*$ defined previously, is an optimal solution of the above optimization (by showing a zero gradient value). Thus for running SMG, we only need to show the existence of an unbiased estimator. Observe that the gradient of $\widetilde{\underline{S}}(w)$ with respect to $w$ can be computed as follows, we have $\nabla \widetilde{\underline{S}}(w)_{2n+1} = 0$, $\forall i \in [n-1]$,
\[
\nabla \widetilde{\underline{S}}(w)_i = -\nabla \widetilde{\underline{S}}(w)_{i + n} =  -\mathbb{E}_{Z \sim \D} \left[\exp\left(-\sum_{j=1}^{n-1}(w_{j} - w_{n + j}) Z_n Z_j - (w_n - w_{2n})Z_n\right) Z_n Z_i\right]
\]
and
\[
\nabla \widetilde{\underline{S}}(w)_n = -\nabla \widetilde{\underline{S}}(w)_{2n} =  -\mathbb{E}_{Z \sim \D} \left[\exp\left(-\sum_{j=1}^{n-1}(w_{j} - w_{n + j}) Z_n Z_j - (w_n - w_{2n})Z_n\right) Z_n\right].
\] 

Here we will focus on only the missing-data model. Similar ideas extend to the random flipped-data model also. Consider the following estimators, $\forall i \in [n-1]$
 \begin{align*}
\underline{g}_{\textsf{miss}}^i(w;X) &= -\frac{\exp\left(-(w_{n} - w_{2n})X_n\right)X_n}{1- p_n}\times\frac{\exp\left(-(w_{i} - w_{n + i})X_nX_i\right)X_i}{1- p_i}\times \\
&\qquad\qquad \prod_{j \neq i,n}\frac{\exp\left(-(w_{j} - w_{n + j})X_nX_j\right) - p_j}{1- p_j}
\end{align*}
and
 \[
\underline{g}_{\textsf{miss}}^n(w;X) = -\frac{\exp\left(-(w_{n} - w_{2n})X_n\right)X_n}{1-p_n} \times \prod_{j \neq n}\frac{\exp\left(-(w_{j} - w_{n- 1 + j})X_nX_j\right) - p_j}{1- p_j}.
\]
We will show that $G(w;X) = \sum_{i=1}^{n} \underline{g}_{\textsf{miss}}^i (e^i - e^{n+i})$ for $e^i \in \mathbb{R}^{2n+1}$ is an indicator vector for coordinate $i$, is an unbiased estimator of the gradient of $\widetilde{\underline{S}}$. Let $v = \Pi^n_{\Delta \rightarrow B}(w)$. Taking expectation of $g_{\textsf{miss}}$ over $\D_{\mathsf{miss}}$, we have for all $i \in [n-1]$
\begin{align*}
\underset{X \sim \D_{\textsf{miss}}}{\E}[\underline{g}_{\textsf{miss}}^i(w;X)] &= -\underset{X\sim \D_{\textsf{miss}}}{\E}\left[\frac{\exp\left(- v_n X_n\right)X_n}{1 - p_n} \times \frac{\exp\left(- v_i X_nX_i\right)X_i}{1 - p_i}\times \prod_{j \neq i,n}\frac{\exp\left(-v_jX_nX_j\right) - p_i}{1- p_i}\right]\\
&= -\underset{\substack{Z \sim \D\\ C_n \sim \mathsf{Ber}(1- p_n)}}{\E}\left[ \frac{\exp\left(-v_nC_nZ_n\right)C_nZ_n}{1-p_n} \times \underset{C_i \sim \mathsf{Ber}(1 - p_i)}{\E} \left[\frac{\exp\left(-v_iC_nZ_nC_iZ_i\right)C_iZ_i}{1 - p_i}\right] \right.\times \\
&\qquad \qquad \left. \prod_{j \neq i,n}\underset{C_j \sim \mathsf{Ber}(1 - p_j)}{\E}\left[\frac{\exp\left(-v_jC_nZ_nC_jZ_j\right) - p_j}{1- p_j}\right]\right]\\
&= - \underset{\substack{Z \sim \D\\ C_n \sim \mathsf{Ber}(1- p_n)}}{\E}\left[\frac{\exp\left(-\sum_{j\neq n}v_jC_nZ_nZ_j - v_nC_nZ_n\right)Z_iC_nZ_n}{1-p_n}\right]\\
&= -\underset{Z \sim \D}{\E}\left[\exp\left(-\sum_{j \neq n}v_jZ_nZ_j - v_nZ_n\right)Z_i Z_n\right] = \nabla \widetilde{\underline{S}}(w)_i.
\end{align*}

Similarly, 
\begin{align*}
&\underset{X \sim \D_{\textsf{miss}}}{\E}[\underline{g}_{\textsf{miss}}^n(w;X)] \\
&= -\underset{X\sim \D_{\textsf{miss}}}{\E}\left[\frac{\exp\left(-v_n X_n\right)X_n}{1-p_n} \prod_{j \neq n}\frac{\exp\left(-v_jX_nX_j\right) - p_j}{1- p_j}\right]\\
&= -\underset{\substack{Z \sim \D\\ C_n \sim \mathsf{Ber}(1- p_n)}}{\E}\left[\frac{\exp\left(-v_nC_nZ_n\right)C_nZ_n}{1 - p_n} \prod_{j \neq n}\underset{C_j \sim \mathsf{Ber}(1 - p_j)}{\E}\left[\frac{\exp\left(-v_jC_nZ_nC_jZ_j\right) - p_j}{1- p_j}\right]\right]\\
&= - \underset{\substack{Z \sim \D\\ C_n \sim \mathsf{Ber}(1- p_n)}}{\E}\left[\frac{\exp\left(-\sum_{j\neq n}v_jC_nZ_nZ_j - v_nC_nZ_n\right)C_nZ_n}{1-p_n}\right]\\
&= -\underset{Z \sim \D}{\E}\left[\exp\left(-\sum_{j \neq n}v_jZ_nZ_j -v_nZ_n\right)Z_n\right] = \nabla \widetilde{\underline{S}}(w)_n.
\end{align*}
This gives us that $\underset{X \sim \D_{\textsf{miss}}}{\E}[G_{\textsf{miss}}(w;X)] = \nabla \widetilde{\underline{S}}(w)$. Lastly we can see that the norm bound of the estimator follows the same proof as the zero-mean field case and we get, $||\underline{G}_{\textsf{miss}}(w;X)||_{\infty} \leq \frac{1}{(1-p_{\max})^{2}} \exp\left(\frac{\lambda}{1-p_{\max}}\right)$ for all $w\in \Delta(\lambda, 2n+1)$. 

Now we can straightforwardly apply Theorem \ref{thm:main} to get the following:
\begin{theorem}
For known constants $p_i$, given samples with missing data from an unknown $n$-variable Ising model $\D(A, \theta)$, for $T \ge \lambda^4\log (n/\delta) \exp(C\lambda)\epsilon^{-4}$ with large enough constant $C>0$, Algorithm \ref{alg:smgm} returns $\bar{w}$ such that:
\[
\forall~ i \in [n-1],~\left|\Pi_{\Delta \rightarrow B}^{n-1}(\bar{w})_i - A_{ni}\right| \leq \epsilon.
\]
\end{theorem}